\documentclass[letterpaper]{article}
\usepackage{aaai}
\usepackage{times}
\usepackage{helvet}
\usepackage{courier}
\usepackage{graphicx}
\usepackage{subfigure}
\usepackage{amsmath}
\usepackage{amsthm}

\usepackage{color}
\usepackage{hyphenat}
\usepackage[vlined, ruled, commentsnumbered, linesnumbered]{algorithm2e}
\newcommand{\algoname}[1]{\textnormal{\textsc{#1}}}
\newtheorem{theorem}{Theorem}

\newtheorem{definition}{Definition}

\SetKw{KwAnd}{and}
\SetKw{KwBreak}{break}
\SetKwInOut{In}{Input}
\SetKwInOut{Out}{Output}

\newcommand*{\LargerCdot}{\raisebox{-0.25ex}{\scalebox{1.2}{$\cdot$}}}

\begin{document}
%
\title{On the k-Anonymization of Time-varying and Multi-layer Social Graphs}
\author{Luca Rossi\\ School of Computer Science\\ University of Birmingham \\ United Kingdom \\ l.rossi@cs.bham.ac.uk \And Mirco Musolesi\\ School of Computer Science\\ University of Birmingham \\ United Kingdom \\ m.musolesi@cs.bham.ac.uk \And
Andrea Torsello\\ DAIS\\ Universit\`{a} Ca' Foscari Venezia\\ Italy \\ torsello@dais.unive.it}%

\maketitle
\begin{abstract}
\begin{quote}
The popularity of online social media platforms provides an unprecedented opportunity to study real-world complex networks of interactions. However, releasing this data to researchers and the public comes at the cost of potentially exposing private and sensitive user information. It has been shown that a naive anonymization of a network by removing the identity of the nodes is not sufficient to preserve users' privacy. In order to deal with malicious attacks, $k$-anonymity solutions have been proposed to partially obfuscate topological information that can be used to infer nodes' identity.

In this paper, we study the problem of ensuring k-anonymity in time-varying graphs, i.e., graphs with a structure that changes over time, and multi-layer graphs, i.e., graphs with multiple types of links. More specifically, we examine the case in which the attacker has access to the degree of the nodes. The goal is to generate a new graph where, given the degree of a node in each (temporal) layer of the graph, such a node remains indistinguishable from other $k-1$ nodes in the graph. In order to achieve this, we find the optimal partitioning of the graph nodes such that the cost of anonymizing the degree information within each group is minimum. We show that this reduces to a special case of a Generalized Assignment Problem, and we propose a simple yet effective algorithm to solve it. Finally, we introduce an iterated linear programming approach to enforce the realizability of the anonymized degree sequences. The efficacy of the method is assessed through an extensive set of experiments on synthetic and real-world graphs.
\end{quote}
\end{abstract}

\section{Introduction}

Interactions between users in an Online Social Network (OSN) can be abstracted using a graph representation. More complex dynamics, such as interactions over time or across multiple media are successfully captured by means of time-varying~\cite{holme2012temporal} and multi-layer networks respectively~\cite{hristova2014_keep}. 
Applications of these datasets include the analysis of structural properties of social networks~\cite{mislove2007measurement}, the investigation of the dynamics of information spreading in social media~\cite{kwak2010twitter}, and the generation of personalized recommendations for online systems~\cite{andersen2008trust}. However, there is an increasing concern for the privacy implications related to the management, mining and distribution of these datasets.

\begin{figure*}[t]
\centering
\includegraphics[width=.75\linewidth]{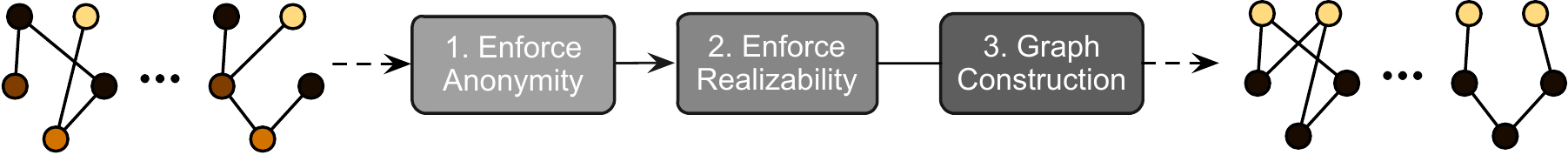}
\caption{The proposed anonymization framework: given a time-varying graph. (1) the temporal degree sequence of each node is anonymized, (2) the resulting degree sequences are checked to ensure that each temporal slice of the anonymized graph is realizable (2), and finally the anonymized time-varying graph is constructed. Here the colors indicate the anonymity groups.}
\label{pipeline}
\end{figure*}

It has been shown that simply generating a random identifier to label the nodes of the graphs does not guarantee privacy~\cite{backstrom2007wherefore}. In fact, an attacker may be able to identify the nodes of a graph simply by collecting information from external sources about their interactions. For example, if the attacker knows that the target individual interacts with a certain number of other users in the network, and that number turns out to be unique for that individual, this piece of information alone is enough to identify the user among all the nodes of the network. 
While the problem of privacy preservation for digital data has been extensively studied in the literature~\cite{meyerson2004complexity,fung2010privacy}, the emergence of large graphs as a tool to model and analyze online social interactions has recently shifted research efforts to the problem of anonymizing structural data~\cite{backstrom2007wherefore,hay2007anonymizing,hay2008resisting,liu2008towards}. In particular, Liu and Terzi~\cite{liu2008towards} provided the first algorithm to guarantee the construction of a $k$-degree anonymous graph. The $k$-anonymity model aims at ensuring that, given a structural query, at least $k$ nodes in the graph satisfy the query. In particular, $k$-degree anonymity guarantees that each node of the graph shares the same degree of at least $k$ other nodes.

Although many networks are naturally modeled as dynamic systems, in most studies the temporal dimension is usually abstracted to produce an aggregated \emph{static} graph. Despite giving an overall picture of the structure which still allows for some interesting analyses, much of the information is lost in the aggregation, and thus researchers have started to turn their attention to the analysis of the dynamic version of the graphs. However, this calls for novel anonymization techniques that are able to cope with the additional longitudinal dimension.

Inspired by the seminal work of Liu and Terzi~\cite{liu2008towards} on anonymization of (single-layer) graphs, we consider the problem of $k$-degree anonymity in time-varying and multi-layer graphs. While most of the algorithms in the literature attempt to solve the $k$-anonymity problem in single-layer graphs, we are interested in protecting the privacy of users participating in a network, which has a structure that evolves with time~\cite{holme2012temporal}, i.e., a time-varying graph, or with multiple type of links associated with the same pair of users, i.e., a multi-layer graph~\cite{mucha2010community}. The structure of a multi-layer graph can be interpreted as that of a time-varying graph where each temporal slice corresponds to a separate layer and the order of the slices does not matter; in the remainder of the paper we will refer exclusively to time-varying graphs for simplicity.

Note that a na\"{\i}ve approach that enforces $k$-degree anonymity independently in each temporal slice is not sufficient to ensure $k$-degree anonymity for the whole time-varying graph, as it is possible to decrease the level of anonymity $k$ by observing the degree sequences of the nodes through time, i.e., their temporal degree sequences. Thus, we need to ensure that the temporal degree sequence of each node is indistinguishable from that of at least $k-1$ other nodes while preserving as much structure of the original time-varying graph as possible. Fig.~\ref{pipeline} shows the pipeline of the proposed approach. Given a time-varying graph and a desired anonymity level $k$ in input, a first module outputs a new set of anonymous degree sequences by solving $l_1$-norm minimization problem using a simple yet effective solution based on a variation of the k-means algorithm~\cite{jain1988algorithms}. A second module ensures these sequences are realizable~\cite{erdos1960graphs}, i.e., that there exists a temporal slice with the given degree sequence, while a third and final module generates a $k$-anonymous time-varying graph from the anonymized and realizable degree sequences.

We conduct an extensive set of experiments on a number of real-world networks, and we show that it is possible to anonymize large graphs while minimizing the loss of structural information. Moreover, we show that when the temporal slices are structurally correlated, i.e., successive slices show a similar structure, the complexity of the anonymization task decreases. To the best of our knowledge, this is the first work to investigate the problem of $k$-degree anonymity in time-varying and multi-layer graphs.


\section{Related Work}\label{related}
The concept of $k$-anonymity in the graph domain~\cite{hay2007anonymizing,hay2008resisting} was introduced by Hay et al., but it is only with Liu and Terzi~\cite{liu2008towards} that a first algorithm to construct a $k$-anonymous graph is proposed. As their algorithm is designed to work on static graphs, however, if applied on the temporal slices of a time-varying graph it fails to take into account the additional information contained in the temporal dimension, i.e., the size of the anonymity groups in the temporal graph will be lower than that of the individual slices. Moreover, their technique generally requires repeated anonymizations of the graph under increasing levels of structural noise, something that is not computationally feasible when dealing with large time-varying graphs. A number of successive works proposed heuristics to reduce the total running time, thus making it feasible to anonymize large static social networks~\cite{lu2012fast,casas2013algorithm,hartung2014improved}.

Chester et al.~\cite{chester2012anonymizing} considered a scenario in which the level of privacy concern of the different nodes of a network varies, i.e., only a subset of nodes of the networks is anonymized. Other researchers, on the other hand, focused on stricter definitions of $k$-anonymity, where the amount of structural information available to the attacker ranges from the immediate neighborhood of a node to the whole graph structure~\cite{hay2008resisting,zhou2008preserving,zou2009k,cheng2010k,zhou2011k}. However, it is worth noting that the more structural information we take into account during the anonymization process, the more noise we need to add to the original graph, and the less informative the resulting anonymized graph will be.

Some researchers have also started investigating the anonymization of time-varying graphs~\cite{zou2009k,bhagat2010prediction,medforth2011privacy}. Zou et al.~\cite{zou2009k} considered the problem of constructing $k$-automorphic graphs, i.e., graphs where each vertex $v$ cannot be distinguished for $k-1$ symmetric vertices given any structural information. The authors also proposed a way to account for graphs that are periodically republished by replacing the nodes IDs with \emph{generalized vertex IDs}. These are designed in a way that makes it impossible for the attacker to link the structural information of nodes across different temporal slices. However, this comes at the cost of being unable to trace a node along the temporal dimension, thus hindering the analysis of the anonymized network. Bhagat et al.~\cite{bhagat2010prediction} extended the list-based anonymization scheme of~\cite{bhagat2009class} to dynamic graphs by grouping the labels of the graph nodes as to maximize nodes and edges anonymity while minimizing structural information loss. Medforth and Wang~\cite{medforth2011privacy} also studied dynamic graphs, but instead of considering a passive attack as in Zou et al.~\cite{zou2009k}, they accounted for an attacker that can actively influence the degree of the target node by interacting with the network. In contrast with these approaches, our aim is to provide a method for anonymizing the node degrees of sequences of graphs so that the label of a node remains fixed during time and the evolution of interactions of specific nodes can be traced and analyzed.


\section{Problem Definition}\label{problemdefinition}

\begin{figure}[t]
\centering
\subfigure[$t=1$]{\includegraphics[width=0.15\linewidth]{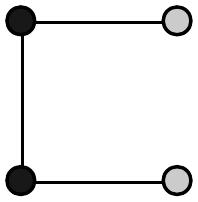}}
\hspace{0.5in}
\subfigure[$t=2$]{\includegraphics[width=0.15\linewidth]{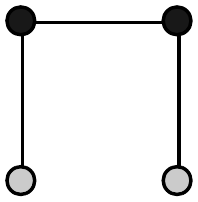}}
\caption{Although each slice $G_t$ satisfies 2-degree anonymity, the temporal degree vectors of the nodes are $[2,2]$,$[2,1]$,$[1,2]$ and $[1,1]$, and thus the time-varying graph does not satisfy 2-degree anonymity.}
\label{toy}
\end{figure}

Let $\mathcal{G} = \lbrace G_1, \cdots, G_T \rbrace$ denote a time-varying graph over a fixed set of vertices $V$, with $|V| = n$. That is, $\mathcal{G}$ is a sequence of undirected and unattributed graphs $G_t(V,E_t)$, with $t=1,\cdots,T$, where $E_t$ is the set of edges active at time $t$. We define the $n \times T$ \emph{degree matrix} $D = \lbrace d_{it} \rbrace$, where $d_{it}$ is the degree of the $i$-th node of $G_t$, and we call the vector $d_{i\LargerCdot} = [d_{i1}, \cdots, d_{iT}]$ the \emph{temporal degree vector}, or temporal degree sequence, of the $i$-th node. Let us also denote with $d_{\LargerCdot t} = [d_{1t}, \cdots, d_{nt}]$ the degree sequence of the $t$-th slice.

Given an arbitrary degree sequence $d_{\LargerCdot t}$, a typical problem is to find a simple graph $G_t$ such that its degree sequence is $d_{\LargerCdot t}$, where an undirected graph is called simple if it has no self loops and has no more than one edge between two vertices. If such a graph exists, the degree sequence is called \emph{realizable}. More formally, Erd\H{o}s and Gallai provide the following necessary and sufficient condition~\cite{erdos1960graphs},
\begin{definition}\label{erdos}
A degree sequence $d_{\LargerCdot t}$, such that $d_{1t} \geq \cdots \geq d_{nt}$, is realizable if and only if $\sum_{i=1}^n d_{it}$ is even and for each $1 \leq j \leq n$ it holds that
\begin{equation}\label{realizable}
\sum_{i=1}^j d_{it} \leq j(j-1) + \sum_{i=j+1}^n \min(d_{it},j)
\end{equation}
\end{definition}
Note that in this paper we will work only with undirected and unattributed graphs, i.e., the adjacency matrix of each $G_t$ is symmetric and binary.
Our approach can be  extended to deal with directed edges by solving two separate anonymization problems, one for the in-degree and one for the out-degree, and ensuring that the resulting degree sequences are realizable~\cite{erdos2010simple}.

\subsection{Temporal Graph Anonymity}

Our goal is to create an anonymized version of a time-varying graph $\mathcal{G}$ such that each node is indistinguishable from $k-1$ other nodes based on its temporal degree vector. Recall that a vector of natural numbers $v$ is \emph{k-anonymous} if, for every entry $v_i$, there exist at least $k-1$ entries $v_j$ with the same value. Based on this definition, Liu and Terzi~\cite{liu2008towards} introduce the concept of a \emph{k-degree anonymous graph}, i.e., a graph whose degree sequence is $k$-anonymous. For example, the vector $v=[1,1,1,2,2]$ is 2-degree anonymous. However, note that the sum of its entries is odd, and thus it does not correspond to any $k$-degree anonymous graph. On the other hand, a complete graph with four nodes is 4-degree anonymous, and its degree sequence is $d_{\LargerCdot t}=[3,3,3,3]$.

Let the $n \times T$ matrix $V$ denote a set of $n$ vectors of length $T$. We say that $V$ is a set of $k$-anonymous vectors if for each row $v_{i \LargerCdot}$, there are at least $k-1$ vectors $v_{j \LargerCdot}$ such that $v_{it}=v_{jt}$, for each $t=1,\cdots,T$. We define a \emph{k-degree anonymous time-varying graph} as follows:
\begin{definition}
A time-varying graph $\mathcal{G}$ is $k$-degree anonymous if its degree matrix $D$ defines a set of $k$-anonymous vectors.
\end{definition}
Note that simply requiring that each temporal slice $G_t$ is $k$-degree anonymous is not sufficient to ensure $k$-degree anonymity in the time-varying graph, as Fig.~\ref{toy} shows.

It should be clear that, independently from the value of $k$ and the structure of $\mathcal{G}$, there always exists a solution to this problem. In fact, the time-varying graph where each $G_t$ is either completely connected or completely disconnected is $k$-degree anonymous, for any $1 \leq k \leq n$. However, such a solution is far from being optimal, in the sense that, in order to anonymize the graph, we need to introduce a large amount of structural noise that inevitably obfuscates the characteristics of the original graph that we aim to preserve. Recall that the edit distance between two graphs is defined as the least-cost edit operations sequence that transforms a graph into another one~\cite{bunke1997relation}. Hence, the optimal solution is to look for the $k$-anonymous graph $\widetilde{\mathcal{G}}$ such that the edit distance between $\mathcal{G}$ and $\widetilde{\mathcal{G}}$ is minimized.


We propose to approximate this problem as follows. We first look for the $k$-degree anonymous degree matrix $\widetilde{D}$ such that the distance
\begin{equation}\label{initial_problem}
dist(D,\widetilde{D}) = \frac{1}{2} \sum_{i=1}^{n} || d_{i \LargerCdot} - \widetilde{d}_{i \LargerCdot} ||_1
\end{equation}
is minimized, where $\widetilde{d}_{i \LargerCdot}$ is the temporal degree sequence of the $i$-th node of $\widetilde{\mathcal{G}}$, and $|| x ||_1$ denotes the $l_1$ norm of the vector $x$. Then, we construct the anonymized graph $\widetilde{\mathcal{G}}$ with degree matrix $\widetilde{D}$ such that the structure of the original graph and its anonymized counterpart are as similar as possible. Note that Eq.~\ref{initial_problem} defines a lower bound on the edit distance, i.e., there is no $\widetilde{\mathcal{G}}$ with degree matrix $\widetilde{D}$ such that its edit distance from $\mathcal{G}$ is smaller than $dist(D,\widetilde{D})$. Thus, we try to minimize the edit distance by first looking for a $k$-anonymous degree matrix $\widetilde{D}$ that is as close as possible to the original one in the $l_1$ sense, i.e., a minimizer of the lower bound, and then building a graph with degree matrix $\widetilde{D}$ such that the edge overlap with the original graph is the largest.

In the next section we show that an optimal solution $\widetilde{D}$ can be found by solving a particular type of generalized assignment problem. However, as we will see, the solution of this problem is not guaranteed to define a set of realizable degree sequences, and thus we will need an additional mechanism to make sure that the anonymized degree sequences of the temporal slices are all realizable. It is worth noting that we allow simultaneous edge additions and deletions. This has been shown to yield a better approximation of the original graph than the case where only edge additions are allowed~\cite{liu2008towards}.

\section{Anonymization Framework}\label{optimization}

Recall that our goal is to partition the graph nodes into groups of size at least $k$, such that each node of a group shares the same temporal degree vector. In addition to this, we want to ensure that a minimal number of structural changes is needed to create these groups. More formally, we are looking for a grouping of the nodes such that the sum of the $l_1$ distances between the temporal degree sequences of the original and the anonymized graph is minimized. This is in general a non-convex and NP-hard problem~\cite{meyerson2004complexity}.

\subsection{Enforcing Anonymity}
We propose to solve an approximation of the above problem based on a variation of the k-means algorithm~\cite{jain1988algorithms} in a $l_1$ metric space. The $k$-means algorithm is a two-step method to cluster points in a $l_2$ metric space. The objective of $k$-means is to minimize the squared deviations from the group centroids, which is equal to the average pairwise squared $l_2$ distances between the points and the centroid of the cluster. Given an initial set of $k$ centroids, the algorithm proceeds by alternating an \emph{assignment step}, where the points are assigned to the closest centroid, and an \emph{update step}, where the new centroids of the clusters are computed. In particular, if the points lie in $l_2$ space the centroid of a cluster is defined as the mean of the points belonging to it. It is important to note that $k$-means can transform a potentially non-convex problem into two convex sub-problems, namely the assignment and the update steps, for which a globally optimal solution can be found.

In contrast with $k$-means, here we need to minimize the average absolute deviation of the original temporal degree sequences from the temporal degree sequences of the anonymized graph. In other words, in our case the centroid of a group is defined as the generalized set median, i.e., the point that minimizes the $l_1$ distance from the points of the group. Let $m = \lceil \frac{n}{k} \rceil$ denote the number of anonymity groups.
We start by defining a random partition of the $n$ nodes into $m$ groups, and we compute the $m \times T$ matrix $P = p_{ij}$ whose rows are the groups medians, i.e., $p_{i \LargerCdot}$ is the set median defined by the $d_{j \LargerCdot}$ assigned to the $i$-th group.

The \emph{assignment step}, on the other hand, requires solving a Generalized Assignment Problem (GAP) with lower bounds~\cite{hamada2011hospitals,krumke2013generalized}. In fact, with respect to the standard version of $k$-means, we have the additional constraint that each group has to hold at least $k$ members in order to guarantee $k$-anonymity. In the classical version of GAP, the goal is to find an optimal assignments of $n$ items to $m$ bins, such that each bin cannot contain more than a fixed amount of items, and the assignment of an item to a bin is associated with a certain cost. 
In our case the size of an item is 1, and the problem is also known as Seminar Assignment Problem (SAP)~\cite{hamada2011hospitals,krumke2013generalized}. Both GAP and SAP are known to be NP-hard~\cite{krumke2013generalized}. However, when the number of bins is fixed, both problems can be written as linear programs and an optimal solution can be found in polynomial time using a standard linear program solver, such as the simplex method or interior point methods. The solution of the linear program assigns $k$ optimal nodes to each group, while the $n-mk$ residual nodes need to be assigned separately. In other words, we are left with an unconstrained assignment problem, where we can assign the residual nodes greedily, i.e., each residual node $i$ to the $l_1$ closest median $j$. We refer to the algorithm solving the assignment step as \algoname{OptimalAssignment}.

Given the optimal assignment, the \emph{update step} consists in calculating the new medians of the clusters. We iterate these steps until convergence, i.e., until the assignment matrix does not change or a user-defined maximum number of iterations is met. Since the algorithm will find a local minimum of the cost function that depends on the initial random partition of the nodes, we repeat the whole procedure a number of times and we select the minimum cost solution. In the remainder of the paper we refer to this as the \algoname{DegreeAnonymization} algorithm. Finally, note that while the original $k$-anonymity problem was non-convex and NP-hard, here we solve two convex sub-problems for which a global optimum can be found.

\subsubsection{Anonymizing Very Large Graphs}
In order to handle large time-varying graphs, for example describing the social interactions between the users of an OSN, we need a fast and efficient way to solve the GAP problem in the assignment step. Krumke and Thielend~\cite{krumke2013generalized} proposed to map GAP to a minimum cost flow problem and using the Enhanced Capacity Scaling algorithm (ECS) to solve it~\cite{krumke2013generalized}. More specifically, the problem of assigning $n$ nodes to $m$ groups can be mapped on a flow network with $|V|=m+n$ nodes and $|E|=mn+m+n$ edges. However, the complexity of the ECS is $O(|E| \log(|V|) (|E|+|V| \log(|V|))$, which makes it unfeasible when applied to very large graphs.

\setlength{\algomargin}{1.3em}
\begin{algorithm}[!t]\label{heuristic}
\small
\BlankLine
\Indm
\In{A degree matrix $D$, a set median matrix $P$ and a desired anonymity $k$}
\Out{An optimal assignment matrix $A$}
\BlankLine
\Indp
\For{$i \leftarrow 0$ \KwTo $l$}
{
	$P \leftarrow$ scramble the rows of $P$\;
	$ A \leftarrow m \times n $ all-zero matrix\;
	\ForEach{$r \in P.rows$}
	{
		$d \leftarrow$ compute distance from $r$ to $D$\;
		$nn \leftarrow$ sort nodes for increasing $d$\;
		\ForEach{$c \in nn$}
		{	
			\If{$k$ nodes have been assigned to $r$}
			{
				\KwBreak\;
			}
			\If{$c$ has not been assigned yet}
			{
				$A(r,c) \leftarrow 1$\;
			}
		}
	}
	$iterA[i] \leftarrow A$\;
}
$A \leftarrow $ select $iterA[i]$ with minimum cost\;
\caption{\algoname{GreedyAssignment}}
\end{algorithm}

We propose to replace the \algoname{OptimalAssignment} algorithm with a less computationally demanding heuristic. The pseudocode of \algoname{GreedyAssignment} is shown in Algorithm~\ref{heuristic}. The algorithm starts by iterating over the set medians, i.e., the rows of $P$, in random order. For each median $r$, it computes the $l_1$ distance from $r$ to the temporal degree vectors in $D$. Then, it assigns to $r$ the first $k$ nodes $c$ that have not been previously assigned to another median. When the anonymity set is complete, i.e., $k$ nodes have been assigned to $r$, the next median is processed. The assignment procedure is repeated $l$ times, each time starting with a different random permutation of the $p_i$s, and the minimum cost assignment is returned. Note that the complexity of our heuristic is $O(l m n\log(n))$, which makes it possible to apply it to very large networks, as opposed to the approach of Krumke and Thielend~\cite{krumke2013generalized}.

\subsection{Enforcing Realizability}

While \algoname{DegreeAnonymization} will return a matrix $\widetilde{D}$ whose columns are $k$-anonymous degree sequences, these are not guaranteed to be realizable.
In Liu and Terzi~\cite{liu2008towards}, when a $k$-anonymous degree sequence is not realizable the authors propose to modify the original graph by adding uniform structural noise in the form of additional edges. The anonymization and noise addition are alternated until a realizable $k$-anonymous degree sequence is returned, while the convergence is guaranteed by noting that, in the worst case, the repeated addition of edges will result in a complete graph, which is by definition $k$-anonymous.

However, while in Liu and Terzi~\cite{liu2008towards} the problem is that of anonymizing a single unattributed graph, in this paper we intend to anonymize a time-varying graph. Not only having $T$ different degree sequences to anonymize there is a higher probability of generating one which is not realizable, but it is also not possible to locally alter the structure of the original temporal slices and the $k$-anonymity group memberships, while ensuring that the $k$-anonymity across the whole time-varying graph is preserved. For this reason, we decide to locally operate on the non-realizable degree sequences in the following way. 

Recall that a degree sequence $\widetilde{d}_{\LargerCdot t}$ is realizable if $\sum_{i=1}^n\widetilde{d}_{it}$ is even and if it satisfies Eq.~\ref{realizable}. Let us first focus on Eq.~\ref{realizable}. Given a temporal slice $\widetilde{G}_t$ and a non-realizable $k$-anonymous degree sequence $\widetilde{d}_{\LargerCdot t}$, we want to project $\widetilde{d}_{\LargerCdot t}$ to the nearest $k$-anonymous degree sequence $\widetilde{d}^*_{\LargerCdot t}$ that satisfies this equation. The function that we want to minimize is the $l_1$ norm between $\widetilde{d}_{\LargerCdot t}$ and $\widetilde{d}^*_{\LargerCdot t}$. That is, our problem can be written as
\begin{equation}\label{initial}
\begin{aligned}
& \text{minimize}
& & || \widetilde{d}^*_{\LargerCdot t} - \widetilde{d}_{\LargerCdot t} ||_1 \\
& \text{subject to}
& & A\widetilde{d}^*_{\LargerCdot t} \leq b(\widetilde{d}^*_{\LargerCdot t}) \\
&&& \widetilde{d}^*_{\LargerCdot t} \geq 0 \\
\end{aligned}
\end{equation}
where $A$ and $b(\widetilde{d}^*_{\LargerCdot t})$ denote respectively the matrix of constraints and the vector of constant terms defined by Eq.~\ref{realizable}, i.e., the $j$-th element of $b(\widetilde{d}^*_{\LargerCdot t})$ is $j(j-1) + \sum_{i=j+1}^n \min(\widetilde{d}^*_{it},j)$. Note, however, that in this formulation of the problem we are allowing all the $n$ components of $\widetilde{d}^*_{\LargerCdot t}$ to vary, and thus the $k$-anonymity of $\widetilde{d}_{\LargerCdot t}$ is not guaranteed to be preserved. Instead, we propose to minimize $|| S\delta_{\LargerCdot t} -S\delta^*_{\LargerCdot t} ||_1$, where $\delta_{\LargerCdot t}$ is the $m$-elements vector such that $\delta_{it}$ is the degree of the nodes in the $i$-th $k$-anonymity group, $m$ denotes the number of groups and $S$ is the $n \times m$ assignment matrix such that the element $S_{ij}=1$ if the $i$-th node belongs to the $j$-th group. 
This can be transformed into a linear program as follows. We first introduce the slack variables $x^+ - x^- = \delta^*_{\LargerCdot t} - \delta_{\LargerCdot t} $, where $x^+,x^- \geq 0$ and we rewrite the objective function as $\mathbf{1}^\top Sx^+ + \mathbf{1}^\top Sx^-$, where $\mathbf{1}$ denotes the all ones-vector. Moreover, recall that $b(\delta^*_{\LargerCdot t})=j(j-1) + \sum_{i=j+1}^n \min(\delta^*_{it},j)$, and thus the constraints are not linear. In order to linearize the constraints, we solve an iterated linear program where we fix an initial value $b(\delta_{\LargerCdot t})$, and we alternate the computation of the optimal $\delta^*_{\LargerCdot t}$ and $b(\delta^*_{\LargerCdot t})$, until convergence. That is, during the $i$-th iteration we solve the linear program
\begin{equation}\label{linear}
\begin{aligned}
& \text{minimize}
& & \mathbf{1}^\top Sx^+ + \mathbf{1}^\top Sx^- \\
& \text{subject to}
& & Ax^+ - Ax^- \leq b(\delta^*_{\LargerCdot t}) - A\delta_{\LargerCdot t} \\
&&& x^+_i, x^-_i \geq 0 \\
\end{aligned}
\end{equation}
where $\delta^*_{\LargerCdot t}$ is the optimal solution at the $(i-1)$-th iteration.

While finding a solution for Eq.~\ref{linear} requires to solve an Integer Linear Program, we propose to solve an alternative problem where the matrix of constraints is totally unimodular and the feasible solutions are guaranteed to be integer-valued. Let us write $A=LS$, where $L$ is the lower triangular matrix, i.e., such that $L_{ij}=1$ if $i \geq j$, $L_{ij}=0$ otherwise. Since $L$ is invertible, we can rewrite the problem in Eq.~\ref{linear} as
\begin{equation}\label{alternative_problem}
\begin{aligned}
& \text{minimize}
& & \mathbf{1}^\top Sx^+ + \mathbf{1}^\top Sx^- \\
& \text{subject to}
& & Sx^+ - Sx^- \leq L^{-1}b(\delta^*_{\LargerCdot t}) - S\delta_{\LargerCdot t} \\
&&& x^+_i, x^-_i \geq 0 \\
\end{aligned}
\end{equation}
\begin{theorem}\label{theorem_realizability}
A solution of the linear program in Eq.~\ref{alternative_problem} satisfies Eq.~\ref{realizable}.
\end{theorem}
\begin{proof}
We need to prove that a feasible solution for the problem in Eq.~\ref{alternative_problem} is also feasible for the problem in Eq.~\ref{linear}. To this end, let us rewrite Eq.~\ref{linear} as
\begin{equation}\label{problem_slack}
\begin{aligned}
& \text{minimize}
& & \mathbf{1}^\top Sx^+ + \mathbf{1}^\top Sx^- \\
& \text{subject to}
& & Sx^+ - Sx^- + L^{-1} z = L^{-1} b(\delta^*_{\LargerCdot t}) - S\delta_{\LargerCdot t} \\
&&& x^+_i, x^-_i, z_i \geq 0 \\
\end{aligned}
\end{equation}
where we introduced the slack variables $z_i \geq 0$. 
By iteratively updating the value of the $x^-_i$s we can ensure that the inequality 
\begin{equation}\label{constraint}
Sx^+ - Sx^- \leq L^{-1}b(\delta^*_{\LargerCdot t}) - S\delta_{\LargerCdot t} 
\end{equation}
holds. As a consequence, we have that
\begin{equation}\label{slack}
L^{-1}z = b(\delta^*_{\LargerCdot t}) - S\delta_{\LargerCdot t} - Sx^+ + Sx^-\geq 0
\end{equation}

In order to show that a solution that satisfies Eq.~\ref{constraint} will also feasible for the problem in Eq.~\ref{problem_slack}, we need to prove that, whenever the former holds, we have $z \geq 0$. Let us introduce a slack variable $y \geq 0$ and rewrite Eq.~\ref{constraint} as
\begin{equation}\label{constraint_slack}
Sx^+ - Sx^- + y = L^{-1}b(\delta^*_{\LargerCdot t}) - S\delta_{\LargerCdot t} 
\end{equation}
From Eqs.~\ref{slack} and \ref{constraint_slack} it follows that
\begin{equation}
y=L^{-1}z \geq 0
\end{equation}
Finally, since $y \geq 0$ and $z=Ly$, $z$ is a sum of non-negative values and thus $z \geq 0$.
\end{proof}
We propose a fast and effective pivot selection algorithm to find a feasible solution for the problem in Eq.~\ref{alternative_problem}. Our goal is that of setting the values of $x^+_i$ and $x^-_i$ as to satisfy the $i$-th inequality constraint, for all $1 \leq i \leq m$. To start, we initialize $x^+$ and $x^-$ as the all-zero vectors. Given the $i$-th constraint, note that we can not have $x^+_i > 0$ and $x^-_i > 0$ at the same time. Thus, when the $i$-th constraint is violated, i.e., $(L^{-1}b(\delta^*_{\LargerCdot t}) - S(\delta_{\LargerCdot t} + x^+ - x^-))_i < 0$, we set $x^-_i$ so that the inequality is reversed, but we let $x^+_i=0$. In other words, we selectively reduce the degree of those groups that violate the constraints. More precisely, the degree of a group is reduced by an amount proportional to the total magnitude of the violated constraints for that group. We then propagate the reduction to the remaining group so as to maintain the order of the degree sequence. We omit the pseudocode of the \algoname{EnforceRealizability} algorithm due to space constraints.

Let $d^*_{\LargerCdot t}$ be the complete degree sequence output by \algoname{EnforceRealizability}. As a last step to ensure that $d^*_{\LargerCdot t}$ is a realizable degree sequence, we need to make sure that $\sum_{i=1}^n d^*_{it}$ is even. To this end, we pick the smallest group with odd degree sum, and we either increase or decrease the degrees of each of its members by 1. More specifically, we pick the operation (increase or decrease) that yields a feasible solution with minimal $l_1$ distance from the original degree sequence.

\begin{figure*}[!t]
\centering
\subfigure{\includegraphics[width=.24\linewidth]{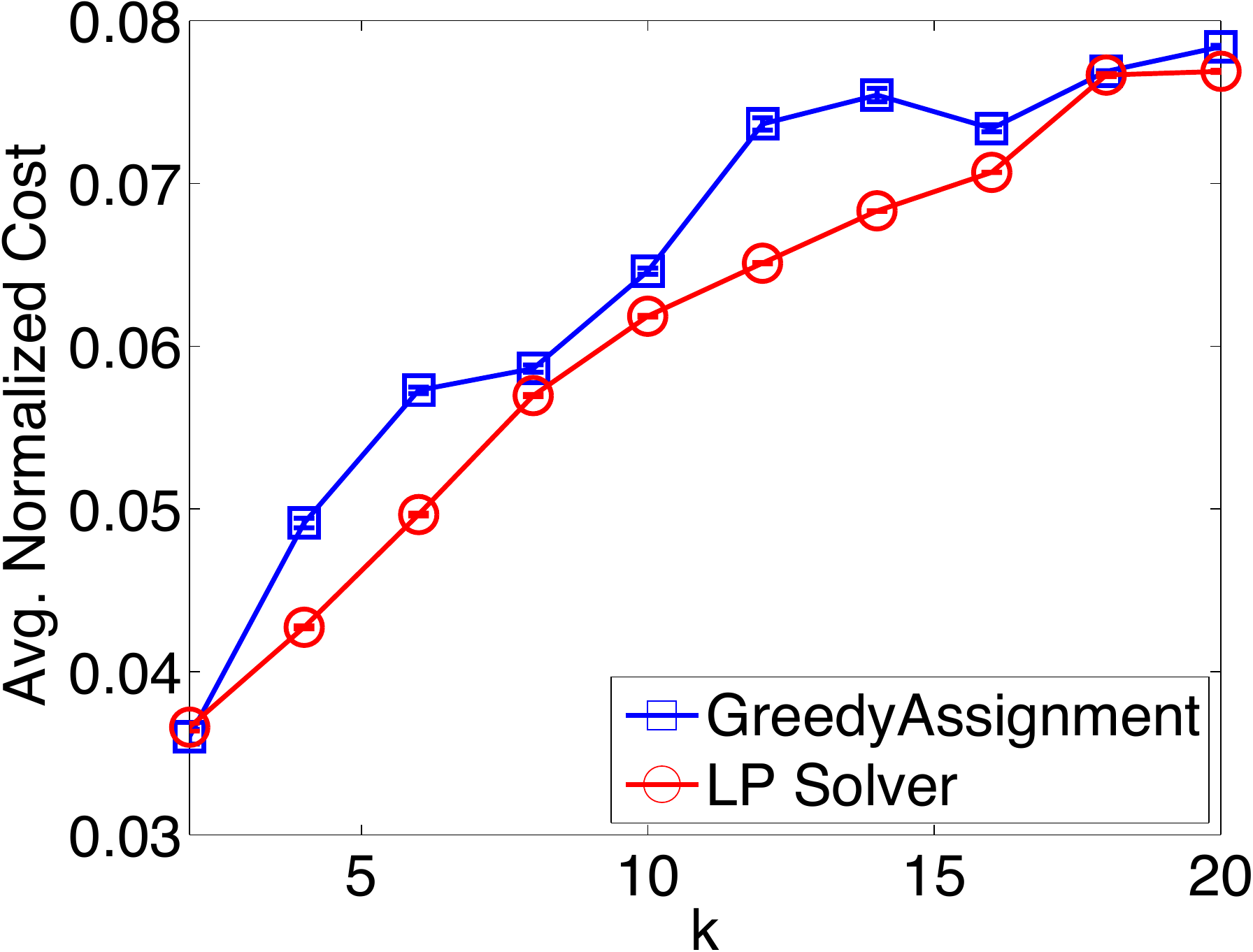}}
\hspace{0.2in}
\subfigure{\includegraphics[width=.232\linewidth]{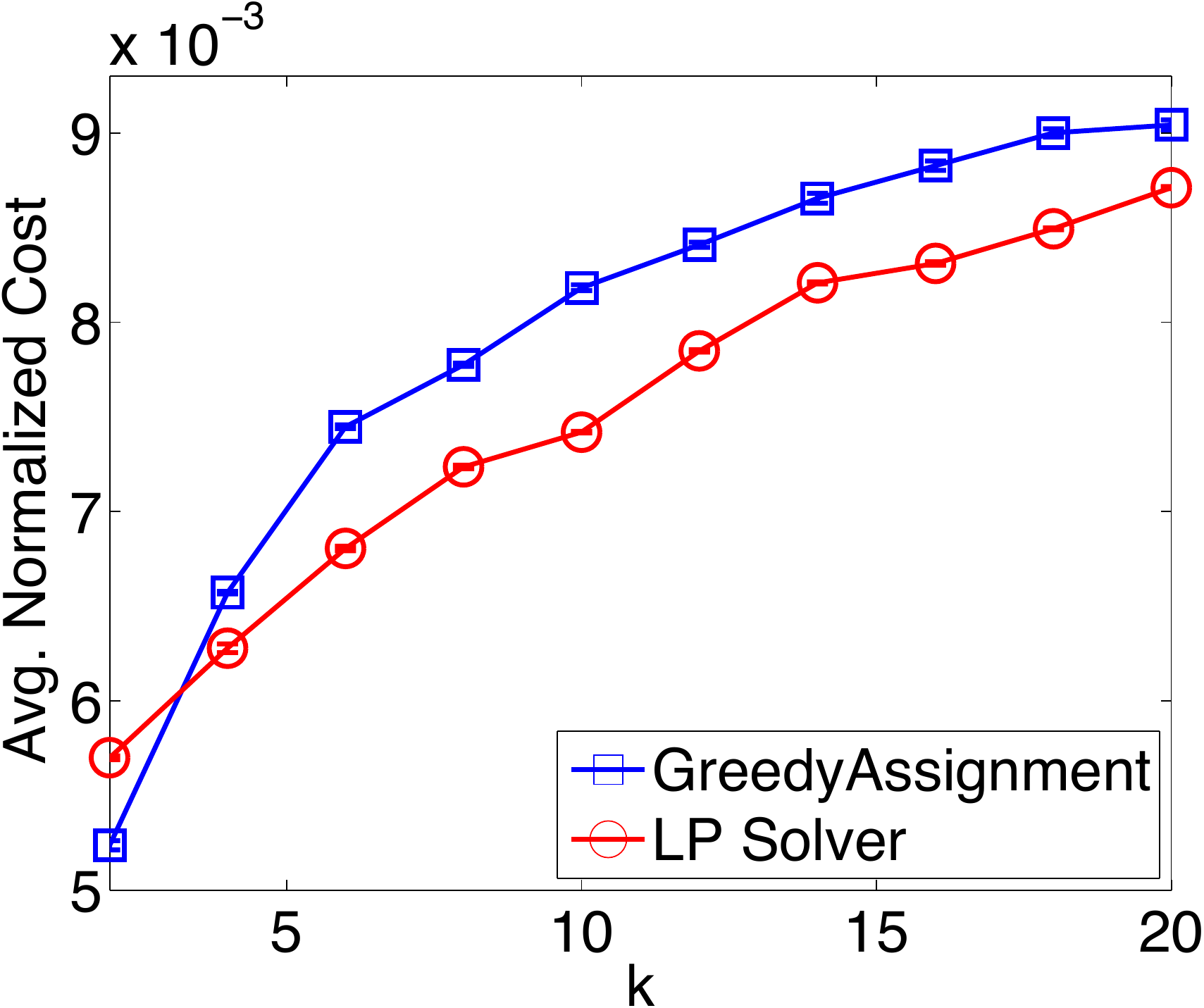}}
\hspace{0.2in}
\subfigure{\includegraphics[width=.244\linewidth]{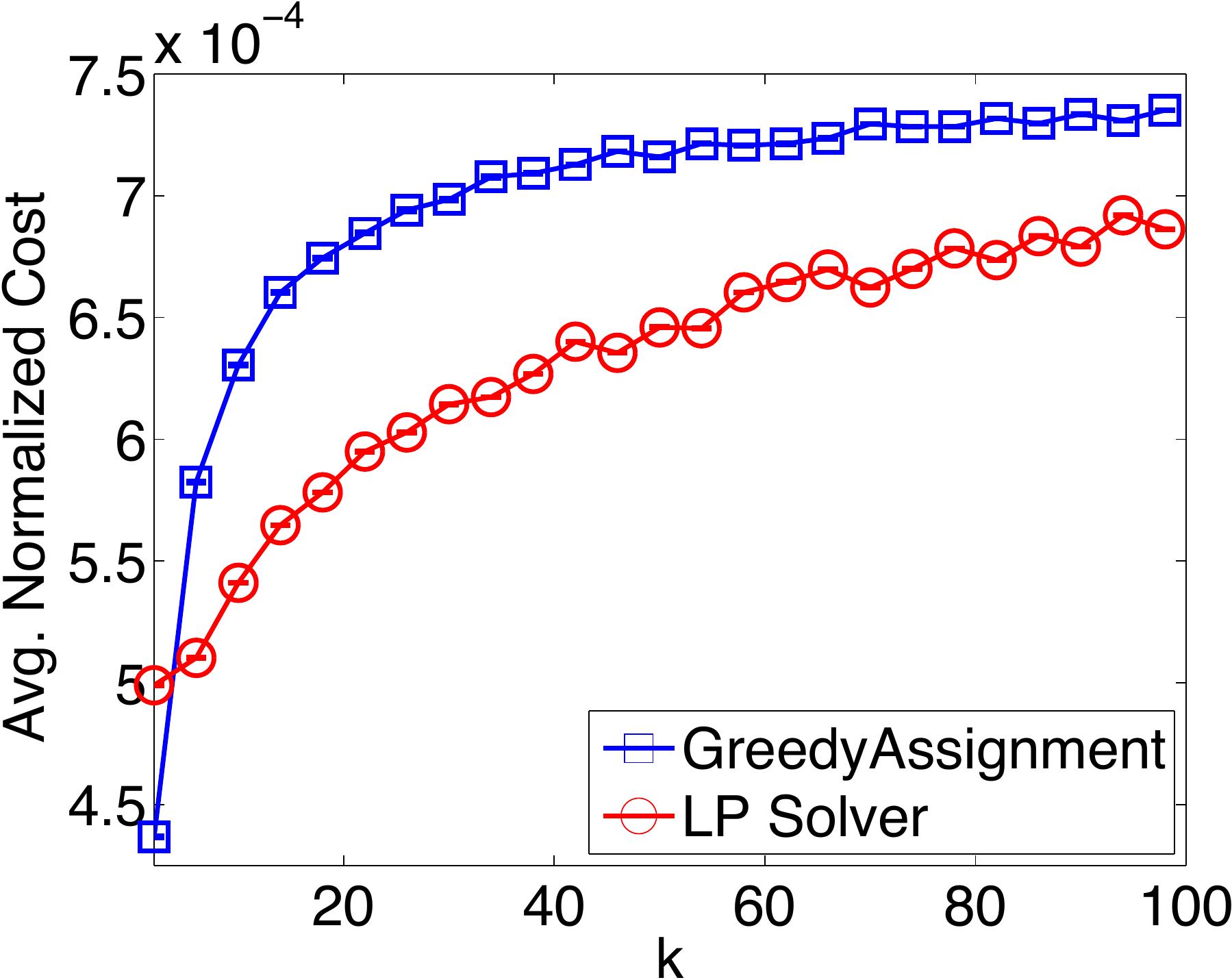}}
\subfigure[MIT]{\includegraphics[width=.24\linewidth]{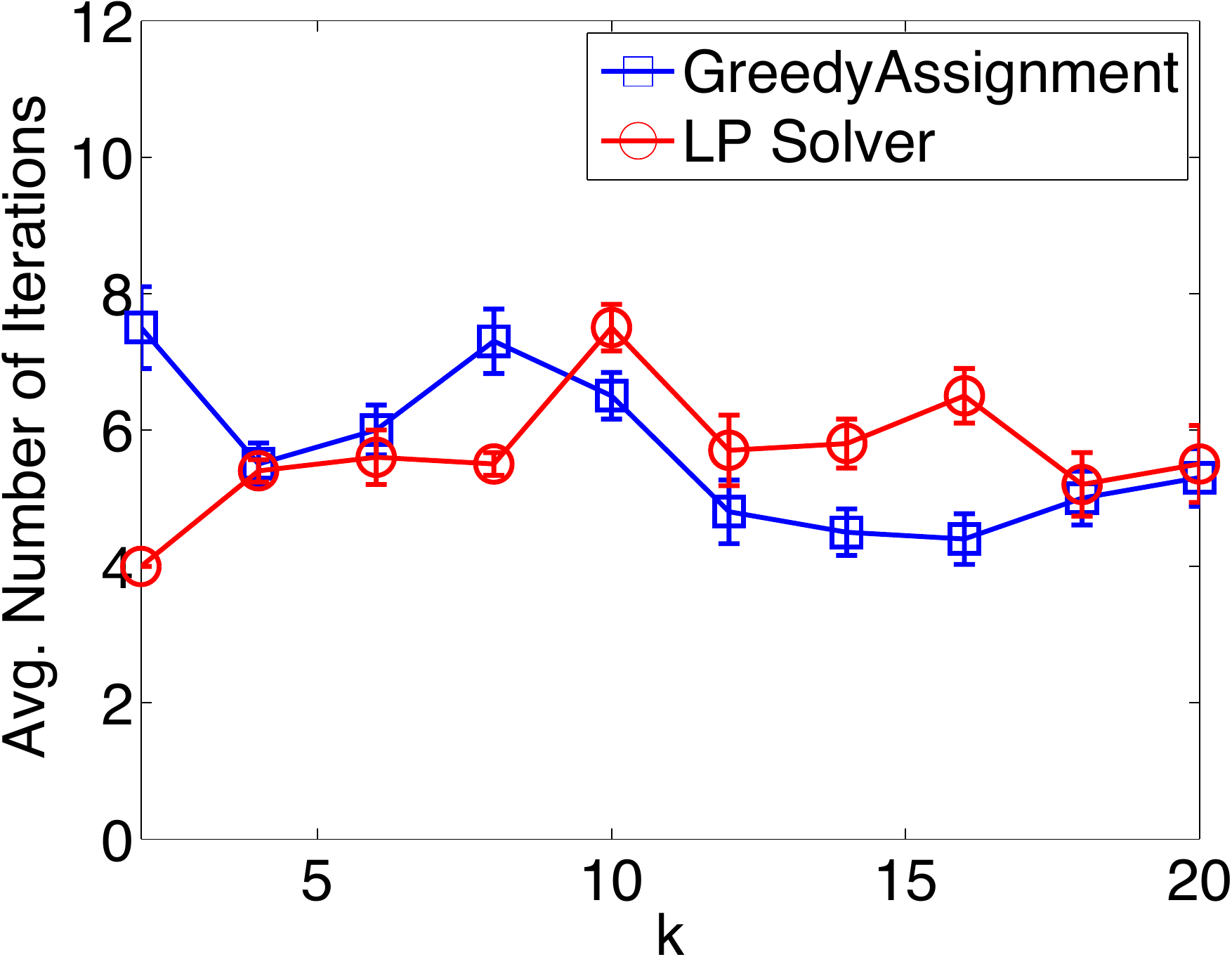}}
\hspace{0.2in}
\subfigure[Enron M]{\includegraphics[width=.24\linewidth]{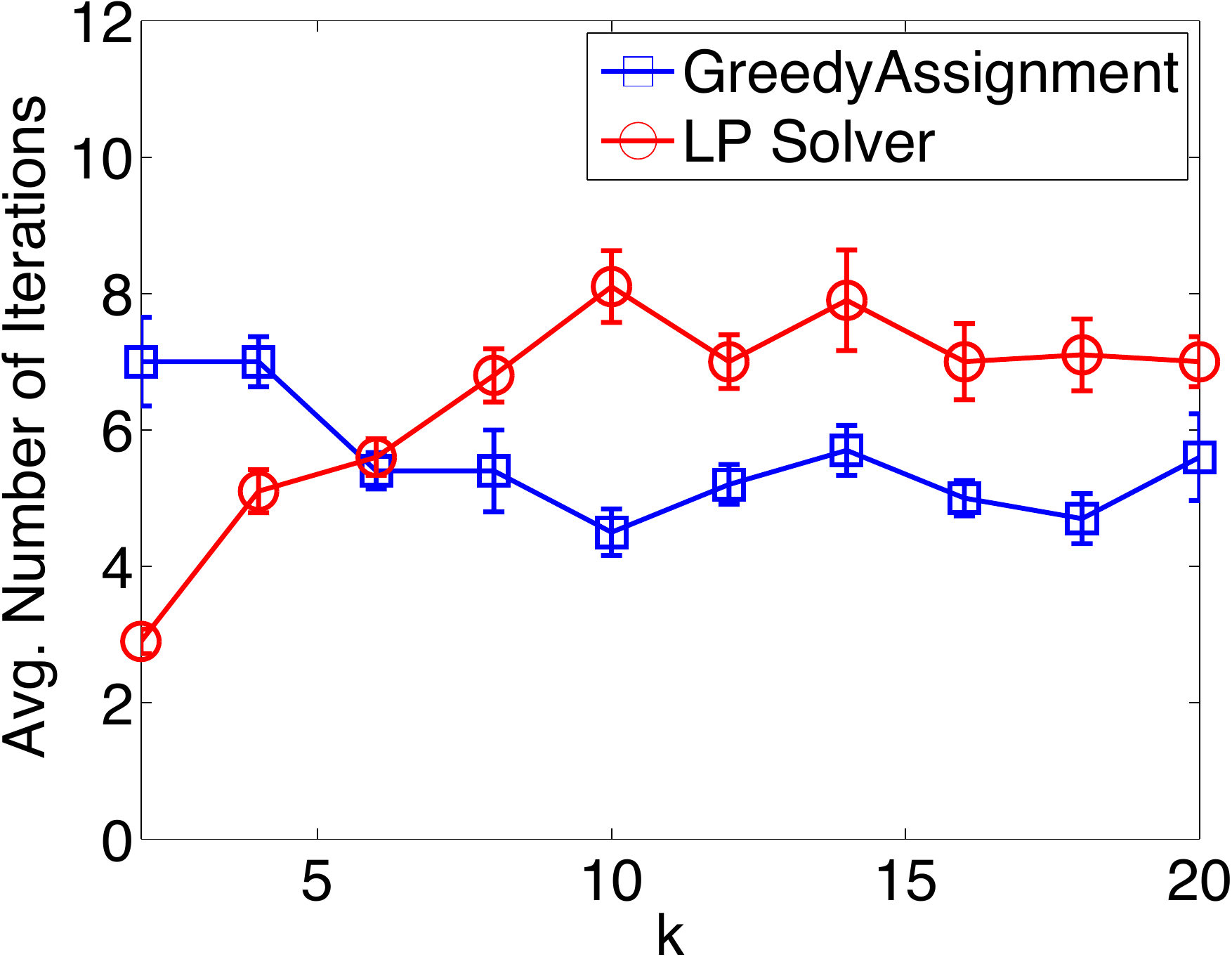}}
\hspace{0.2in}
\subfigure[Irvine]{\includegraphics[width=.24\linewidth]{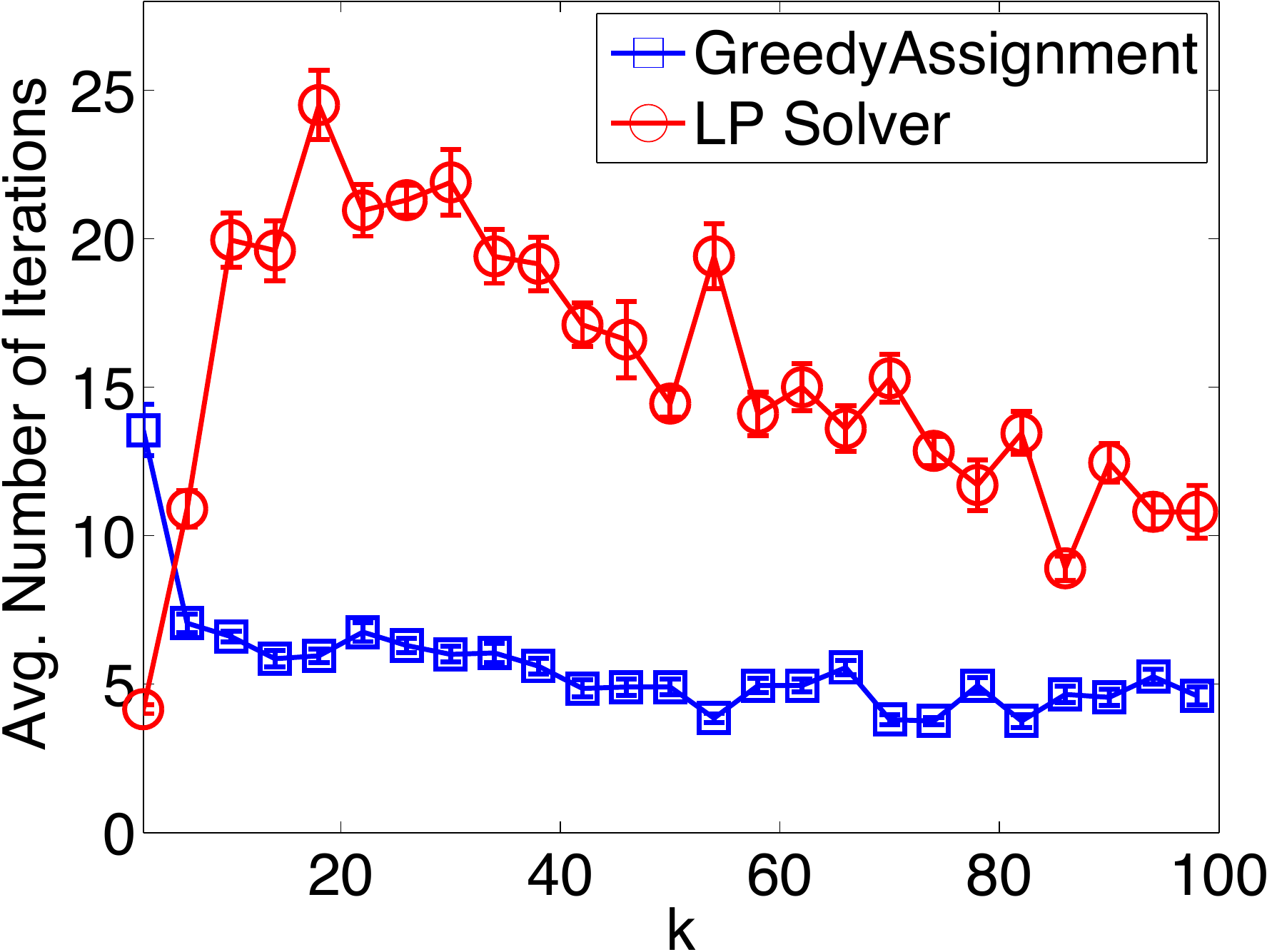}}
\vspace{-0.1in}
\caption{Average anonymization cost (top) and average number of iterations (bottom) needed to converge by \algoname{DegreeAnonymization} using \algoname{GreedyAssignment} and the LP solver, for varying values of $k$.}
\label{greedy_anonymization_iter}
\end{figure*}

\subsection{Graph Construction}\label{construction}
With the anonymous and realizable degree matrix to hand, we can proceed to construct the anonymized time-varying graph. More specifically, we build each temporal slice independently using the relaxed graph construction method of Liu and Terzi~\cite{liu2008towards}. Given a realizable degree sequence $d^*_{\LargerCdot t}$, we use the \algoname{Priority} algorithm to build a graph $\widetilde{G_t}$, such that the edge overlap with the original graph is as large as possible. The \algoname{Priority} algorithm creates a degree anonymous graph with a high edge intersection with the original graph by prioritizing the construction of edges between vertices that share an edge in the original graph. Note that in~\cite{liu2008towards} when the input sequence is not realizable the \algoname{Priority} algorithm needs to call the \algoname{Probing} procedure. This in turn perturbates the structure of the original graph and calls again \algoname{Priority} and attempts to create a new anonymous and realizable degree sequence. Here, instead, we are guaranteed that the input sequence is realizable, and thus there is no need to run \algoname{DegreeAnonymization} again.

\section{Experiments}\label{experiments}
We evaluated our framework on a number of real-world and synthetic datasets. As the final graph is constructed using the graph construction method of Liu and Terzi~\cite{liu2008towards}, we focus most of our evaluation on the heuristics proposed to compute the anonymous matrix $\widetilde{D}$ from to the original matrix $D$. More specifically, the quality of our solution is measured in terms of the normalized cost $\mathcal{C}(D,\widetilde{D})=\sum_{i=1}^{n} \sum_{t=1}^{T}\frac{|d_{it} - \widetilde{d}_{it}|}{T n (n-1)}$, and the results are presented in terms of average normalized cost ($\pm$ standard error) over $20$ repetitions. Also, recall that both \algoname{DegreeAnonymization} and \algoname{GreedyAssignment} depend on a number of parameters. Unless otherwise stated, we set the number of iterations $l$ in Algorithm~\ref{heuristic} to $10$ and we allow a maximum of $50$ iterations before convergence in \algoname{DegreeAnonymization}.


%
%

\subsection{Data}

The \textbf{MIT Social Evolution} dataset~\cite{madan2012sensing} 
consists of 5 layers representing different types of social connections between 84 students, for a total of $7,055$ edges. The layers represent respectively the connections between: 1) close friends, 2) students that participated in at least two common activities per week, 3) discussed politics at least once since the last survey, 4) shared Facebook photos, 5) shared blog/Live Journal/Twitter activities. 

The \textbf{Enron} dataset~\cite{shetty2005discovering} consists of the time-varying network representing e-mail exchanges between 151 users during the period from May 1999 to June 2002 (1,146 days). We consider three alternative versions of this graph, where the slices represent the activity over a month, week and day, respectively. We will refer to these three graphs as \textbf{Enron M}onth ($7,277$ edges over 38 months), \textbf{Enron W}eek ($13,080$ edges over 164 weeks) and \textbf{Enron D}ay ($21,257$ edges over $1,146$ days).

The \textbf{Irvine} dataset~\cite{opsahl2009clustering} represents the social connections between $1,899$ users of an online students community at University of California, Irvine. The data consists of $20,296$ interactions over a period of 51 days.

Finally, the \textbf{Yahoo} dataset\footnote{http://sandbox.yahoo.net} is a collection of 28 days of Yahoo Instant Messenger events, where each node is an IM user, and each link represents a communication event on a given day. This is the largest dataset considered in our study, with a total of 100,000 nodes interacting over 28 temporal slices. We consider also a reduced version of this graph where we select $10,000$ through a bread-first exploration of the largest connected component of the aggregated graph over the 28 days. We refer to the two versions of the Yahoo graph as \textbf{Yahoo $\boldsymbol{10^4}$} ($139,524$ edges) and \textbf{Yahoo $\boldsymbol{10^5}$} ($2,026,734$ edges).


In addition to these real-world datasets, we add a set of synthetic time-varying graphs where the temporal duration of an edge is sampled from a geometric distribution with parameter $\theta$, i.e., $\theta$ represents the probability of an edge to change from being absent (present) to being present (absent). By varying $\theta$ we can control the average temporal correlation~\cite{clauset2012persistence} of the time-varying graph, where the temporal correlation of a graph measures the overlap between successive temporal realizations of the nodes neighborhoods. Thus, sampling from a geometric distribution with a high $\theta$ will generate unstable time-varying graphs where edges constantly appear and disappear, i.e., graphs with a low temporal correlation. On the other hand, when $\theta \rightarrow 0$ the probability of a structural change is close to zero, i.e., the time-varying graph shows a high temporal correlation. Given a value of $\theta$, we then generate a time-varying graph by sampling $10$ temporal slices over $100$ nodes. In total, we generate 11 time-varying graphs with increasing average temporal correlation.


\begin{figure*}[!t]
\centering
\subfigure[Realizability]{\includegraphics[width=.25\linewidth]{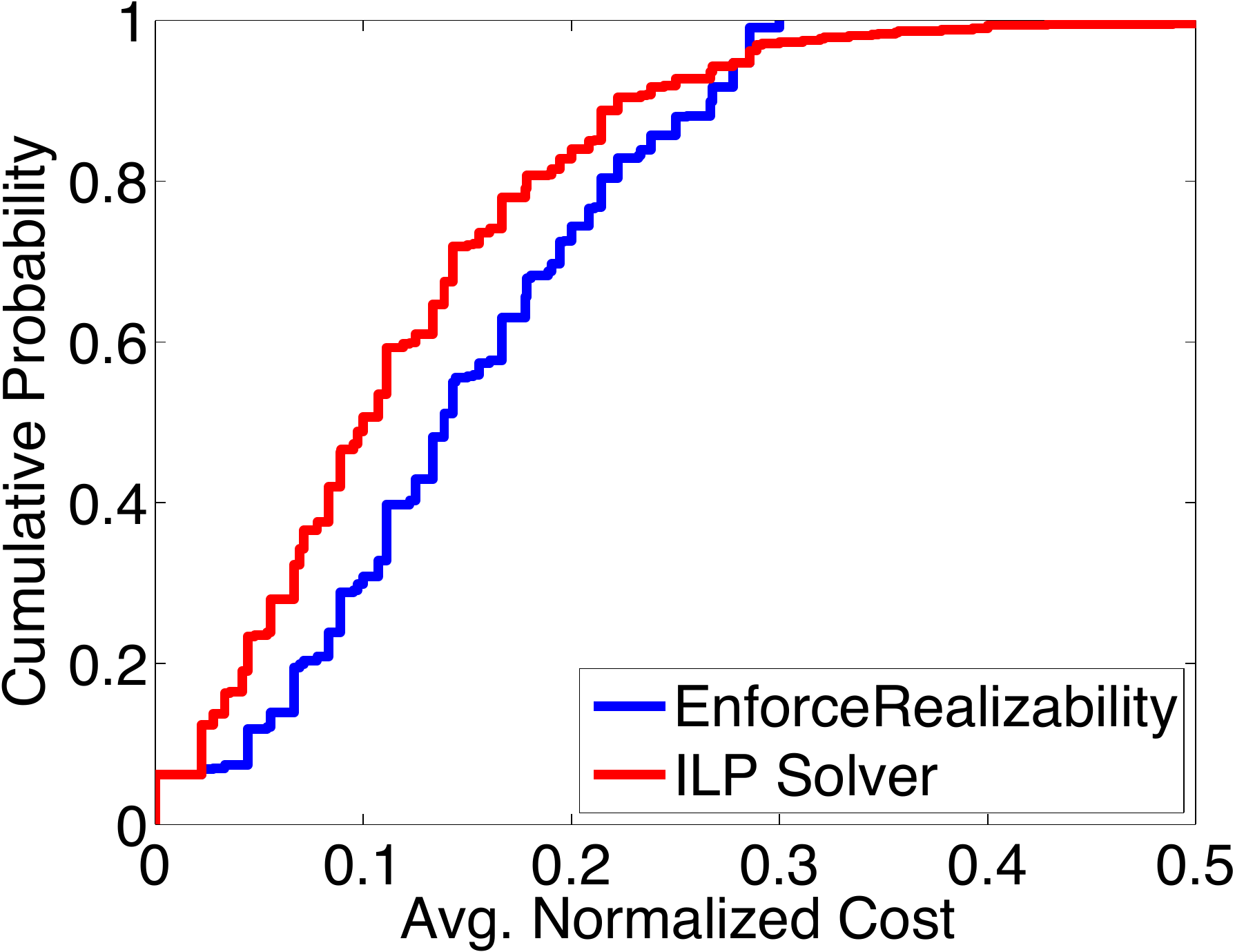}\label{realizability}}
\hspace{0.2in}
\subfigure[Temporal Correlation]{\includegraphics[width=.25\linewidth]{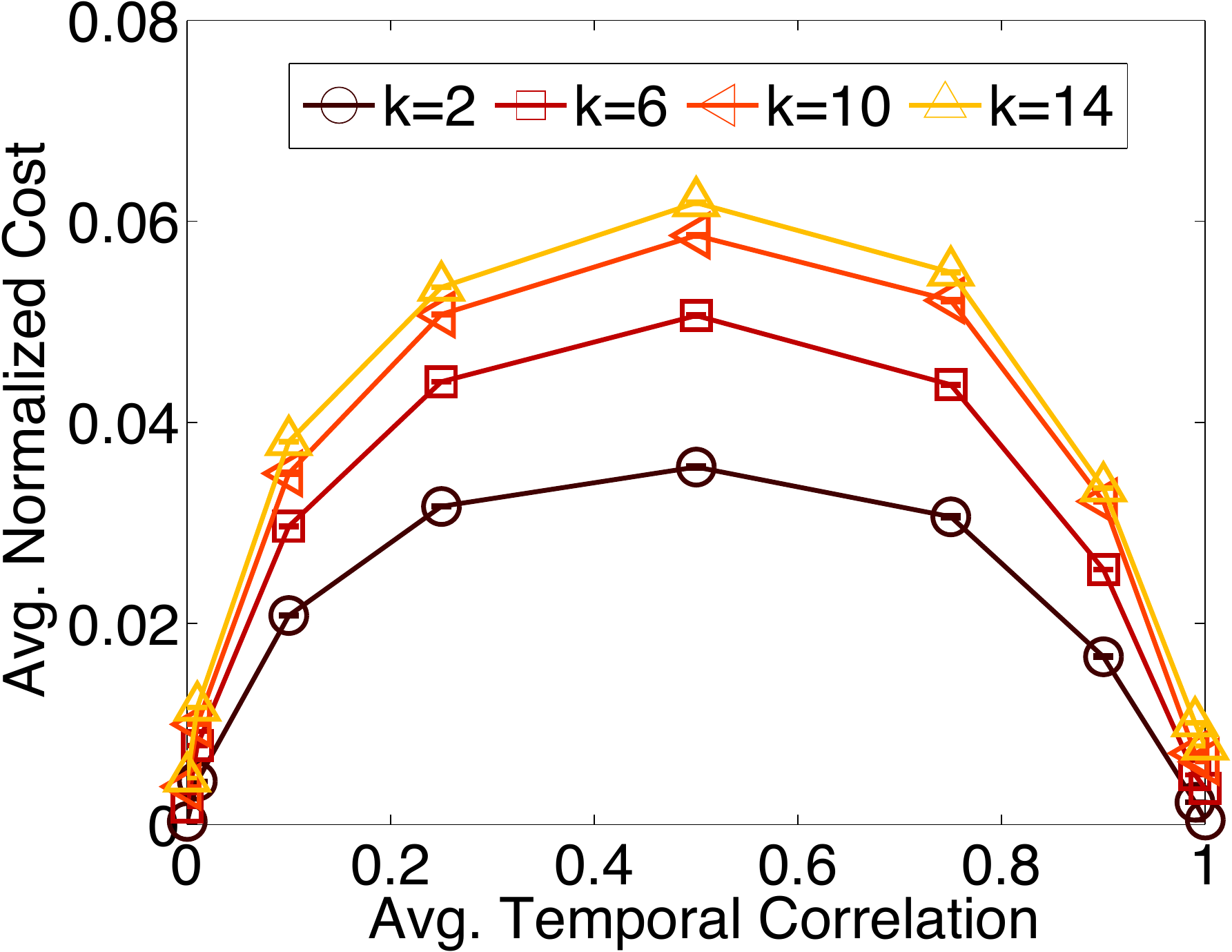}\label{correlation}}
\hspace{0.2in}
\subfigure[Temporal Resolution]{\includegraphics[width=.25\linewidth]{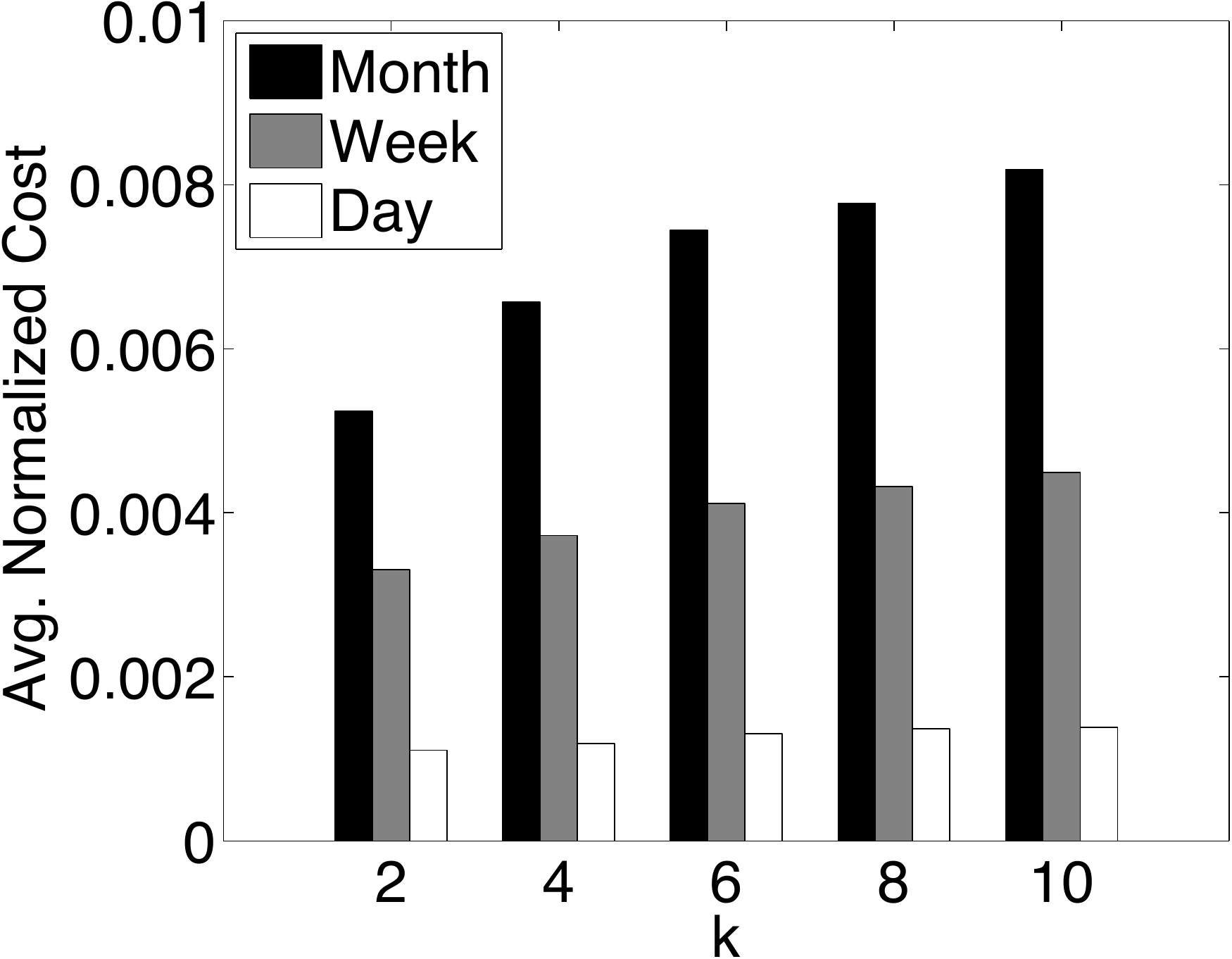}\label{granularity}}
\vspace{-0.1in}
\caption{(a) Empirical cumulative distribution function of the average anonymization costs for $1,000$ non-realizable degree sequences. (b) Average anonymization cost on the synthetic datasets for increasing levels of average temporal correlation. (c) Average anonymization cost on the Enron Month dataset for varying values of $k$ and varying temporal resolution.}
\end{figure*}

\subsection{Degree Anonymization}

As a first experiment, we evaluate the efficiency of \algoname{DegreeAnonymization} and in particular of our \algoname{GreedyAssignment} heuristic. We compare it with an optimal assignment of the nodes to the anonymity groups obtained by solving the assignment step with a standard LP solver. In these experiments we use the revised simplex algorithm implemented in the GNU Linear Programming Kit (GLPK), version 4.35\footnote{http://www.gnu.org/software/glpk/glpk.html}, and we compare the solution found by the LP solver to that of \algoname{GreedyAssignment}. GLPK is a free and open source software that is designed to solve large-scale linear programs. However, note that in these experiments we make use only of the MIT, Enron and Irvine datasets, as the LP solver was not able to scale to the size of the Yahoo datasets. Note in fact that in this dataset when $k=2$ the number of anonymity groups is $m=50,000$, for a total of 5 billion possible node-to-group assignments, i.e., the matrix of coefficients in the GAP has 5 billion entries.

Fig.~\ref{greedy_anonymization_iter} shows the average normalized cost and the average number of iterations to convergence for \algoname{DegreeAnonymization}, when either \algoname{GreedyAssignment} or the LP solver are used to determine the optimal assignment. Our first observation is that in general in both cases as $k$ increases the anonymization cost increases. This is not unexpected, as creating larger anonymity groups requires introducing an increasing amount of structural noise. However, we also note a slight drop in the cost for some values $k$. This may be linked to the cost of assigning residual nodes through the \algoname{AssignResidual} procedure. Recall, in fact, that the number of residual nodes depends on the anonymity level $k$, and it can increase for some $k_2>k_1$. As an example, let us consider a graph with 12 nodes, where with $k=4$ there are no residual nodes to assign, whereas with $k=3$ there is 1 node left to assign.

Fig.~\ref{greedy_anonymization_iter} (top) also shows that our \algoname{GreedyAssignment} heuristic generally achieves a good approximation of the anonymization cost with respect to the LP solver. In particular, for $k=2$ our heuristic consistently outperforms the optimal solution found using the LP solver. This in turn may be related to the problem of local minima for the simplex method. On the other hand, our greedy exploration of the function landscape seems to lead \algoname{DegreeAnonymization} to find a better local minimum. However, as the landscape gets more complicated, i.e., as $k$ increases, the performance of our heuristic quickly deteriorates and \algoname{DegreeAnonymization} with the \algoname{GreedyAssignment} heuristic performs consistently worse than with the LP solver, although the two costs remain generally close.

Fig.~\ref{greedy_anonymization_iter} (bottom) shows the average number of iterations needed to converge for \algoname{DegreeAnonymization}, when either \algoname{GreedyAssignment} or the LP solver are used to determine the optimal assignment. Interestingly, over all the 3 datasets for $k=2$ we see that \algoname{DegreeAnonymization} with the LP solver quickly converges to a non-optimal local minimum, while \algoname{GreedyAssignment} leads to a slower convergence and a better local minimum. As $k$ increases, \algoname{DegreeAnonymization} tends to reach convergence in fewer iterations when \algoname{GreedyAssignment} is used. Although Fig.~\ref{greedy_anonymization_iter} shows that this leads to a slightly higher anonymization cost, Fig.~\ref{greedy_anonymization_iter} shows that this is compensated for by a faster convergence. Moreover, it should be noted that a single run of \algoname{GreedyAssignment} is considerably faster than a single run of the LP solver.

\begin{figure*}[!t]
\centering
\subfigure[Enron M]{\includegraphics[width=.25\linewidth]{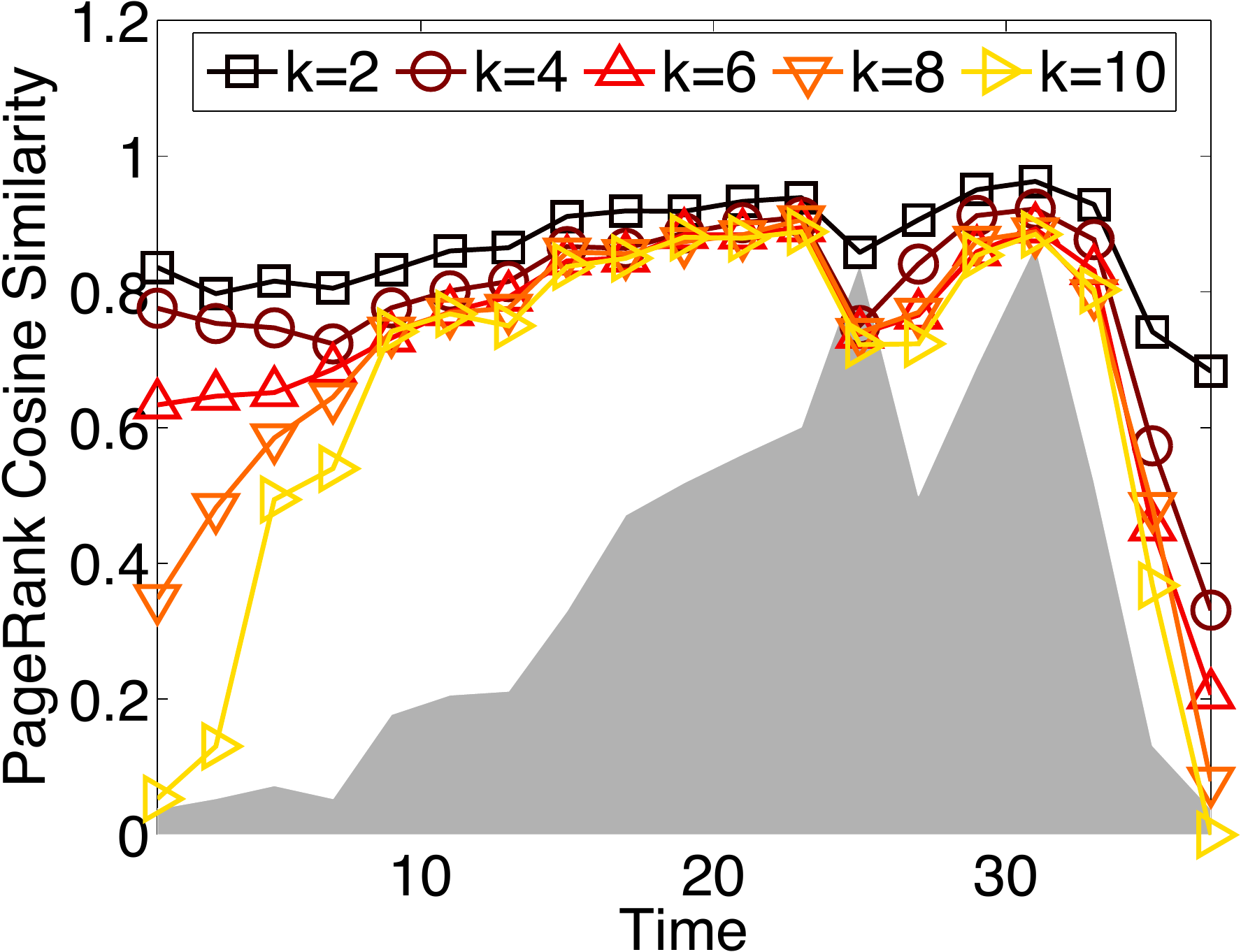}}
\hspace{0.2in}
\subfigure[Irvine]{\includegraphics[width=.25\linewidth]{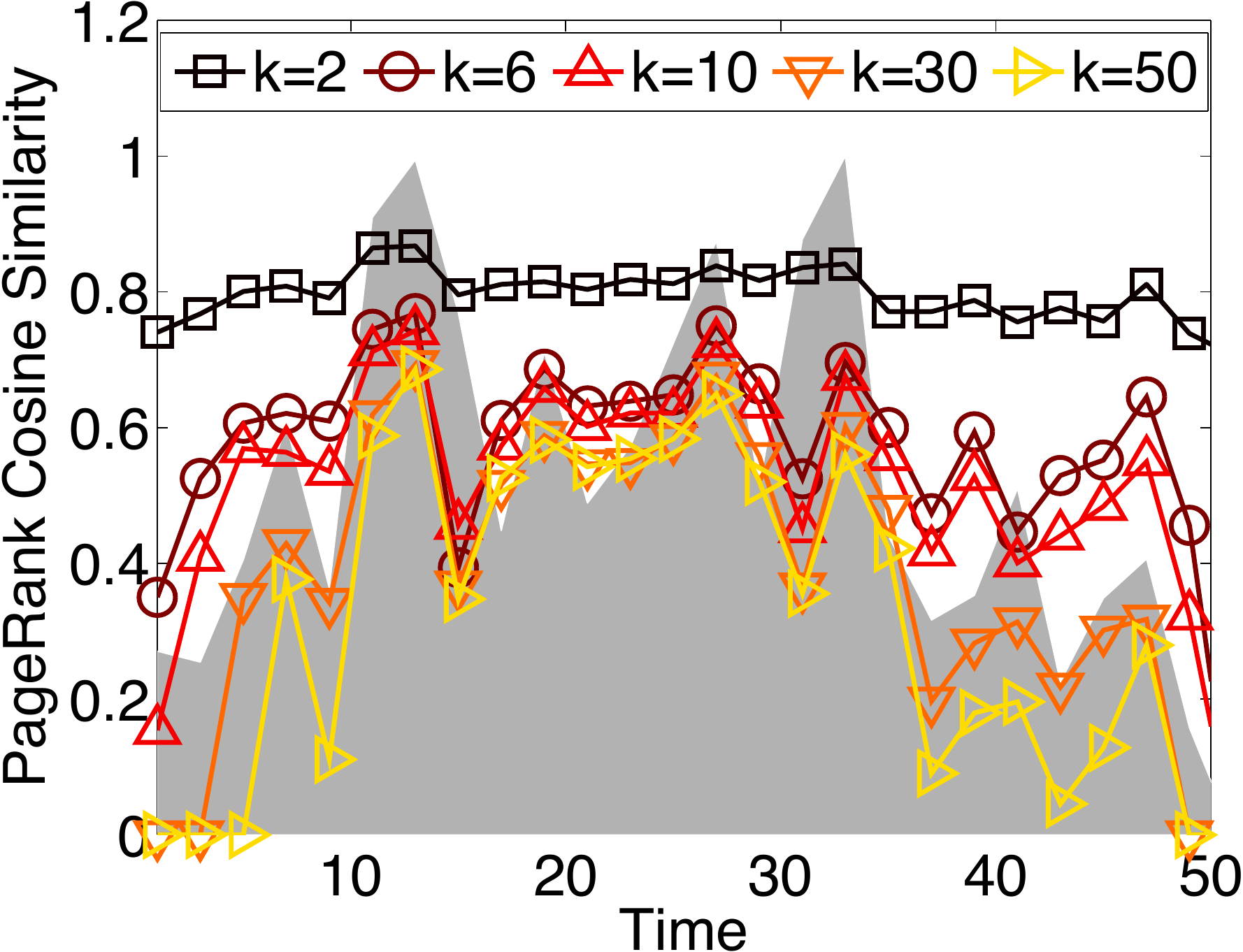}}
\hspace{0.2in}
\subfigure[Yahoo $10^5$]{\includegraphics[width=.25\linewidth]{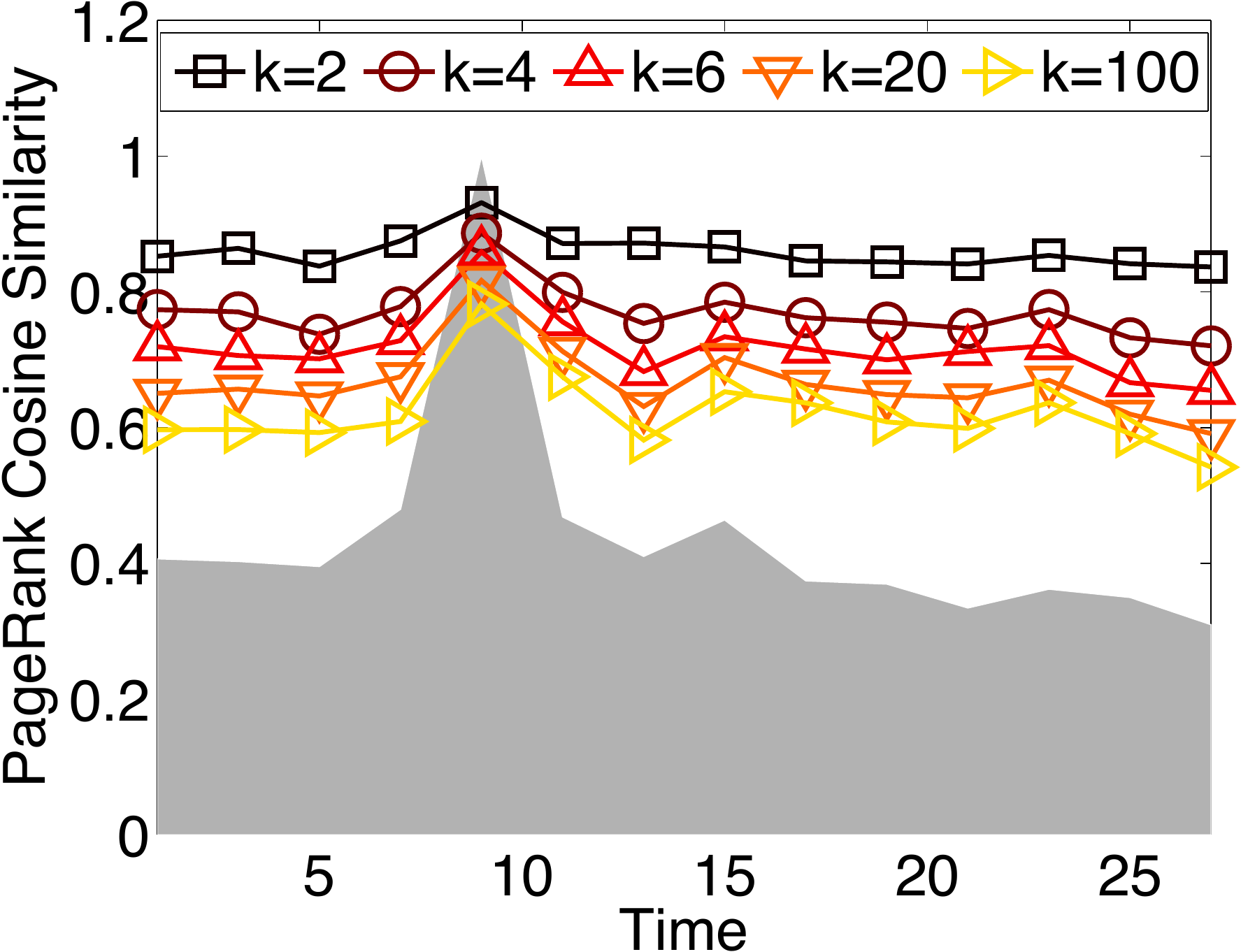}}
\vspace{-0.1in}
\caption{Average cosine similarity between the PageRank vectors of the original temporal slices and the anonymized ones. The shaded area shows how the number of active edges varies with time.}
\label{pr}
\end{figure*}


\subsection{Degree Sequences Realizability}

We now evaluate the \algoname{EnforceRealizability} algorithm. More specifically, we are interested in comparing the solutions obtained solving Eq.~\ref{linear} with those of Eq.~\ref{alternative_problem}. Recall in fact that the feasible solutions of Eq.~\ref{alternative_problem}, although feasible also for Eq.~\ref{linear}, are only a subset of those. 
We generate 1,000 random degree sequences over 10 nodes, where each sequence is created such that it is $k$-anonymous but not realizable, for a random level of anonymity $k$. The reason why we resort to synthetic data is that in our experiments we observe that, when we apply \algoname{DegreeAnonymization} on real-world data, the anonymized degree sequence of each slice is almost always realizable, a behaviour that was also observed in~\cite{liu2008towards}\footnote{The sum of the degree sequence may still be odd, but this can be fixed easily without invoking \algoname{EnforceRealizability}.}. This may be due to the fact that the original degree sequences are indeed realizable, and thus it may be more likely that the anonymized degree sequences are also realizable.

%
%

Fig.~\ref{realizability} shows the empirical cumulative distribution function of the normalized costs of the degree sequences obtained by solving the Integer Linear Program of Eq.~\ref{linear} and the linear program of Eq.~\ref{alternative_problem}, where the solution of Eq.~\ref{alternative_problem} is computed using the \algoname{EnforceRealizability} algorithm. It is interesting to note that the solutions found by \algoname{EnforceRealizability} are close to those found by solving the original NP-hard problem. 

\subsection{Temporal Correlation}

The average temporal correlation is defined between 0 and 1~\cite{clauset2012persistence}, where a value of 0 is achieved for completely anti-correlated temporal slices, a value of 1 is achieved for completely correlated temporal slices, while non-correlation corresponds to a 0.5 value. Fig.~\ref{correlation} shows that, independently of $k$, when the temporal slices are strongly correlated the anonymization costs drops dramatically. In fact, the more homogeneous the structure of the slices is, the easier it is to define anonymity groups that will remain consistent on all the slices without introducing much noise. Note that also a strongly anti-correlated time-varying graph can be considered homogeneous, as the in the limit case, i.e., $p=1$, the structure of the graph at time $t$ is always identical to that at time $t\pm2$.

\subsection{Temporal Resolution}

Similarly to the temporal correlation, the temporal resolution of the graph slices may also influence the complexity of the anonymization task. In fact, Fig.~\ref{granularity} shows that, as we increase the temporal resolution, the average anonymization cost becomes less dependent on $k$. More specifically, with a temporal resolution of 1 month, the cost of enforcing $k=10$ anonymity is about $50\%$ higher than for $k=2$, whereas with a temporal resolution if 1 week there is a $35\%$ increase and with a 1 day resolution a $25\%$ increase. In general, we can think that the higher the temporal resolution, the sparser the slices will be. A sparse graph is naturally anonymous, as a large number of nodes are completely disconnected from the rest, i.e., they remain idle, and thus they are indistinguishable from each other. Thus, when enforcing anonymity across the temporal dimension, we have a higher degree of freedom when grouping the nodes in sparse slices while minimizing edit operations. In the limit case, if we consider a small enough time interval, some temporal slices become empty, i.e., we observe no interactions between the nodes during some periods. Here each node is indistinguishable from the remaining $n-1$ nodes of the graph, and no structural alteration is needed in these slices when aligning the anonymity groups across the longitudinal dimension.

\subsection{Impact on Graph Structure}

So far we have been evaluating our anonymization framework in terms of the normalized anonymization cost. However, this gives us only a partial insight on the information loss that we incur when we anonymize a time-varying graph. In order to evaluate better the effects of the structural perturbation caused by the anonymization process, we evaluate the page PageRank~\cite{page1999pagerank} of the anonymized time-varying graph. The PageRank vector is a measure of node importance commonly used in network analysis. We compute the PageRank vector of each temporal slice for both the original and the anonymized graphs. Fig.~\ref{pr} shows the cosine similarity~\cite{jain1988algorithms} between the PageRank vectors of the original and anonymized temporal slices as a function of time, over three different datasets. Note the shaded area showing the varying volume of interactions over time. When $k$ is low, the PageRank centrality of the nodes is well approximated in the anonymized graph, suggesting that most of the structural information is retained. However, as the level of anonymity increases, more noise needs to be added and the centrality of the anonymized nodes starts to deviate from its original value. Interestingly, we observe that the cosine similarity remains high on the temporal slices where most of the interactions are concentrated. This should not come as a surprise, as sparser slices are more sensitive to the addition and removal of edges.

\subsection{Runtime Evaluation}

\begin{table}[!t]
\scriptsize
\centering
\vspace{0pt}
\begin{tabular}{|l||c||c||c|}

\hline
~Dataset~&~$k=2$~&~$k=5$~&~$k=10$\\ \hline \hline

~MIT~&~$0.227 \pm 0.009$~&~$0.138 \pm 0.007$~&~$0.131 \pm 0.006$~\\
~Enron (M)~&~$0.578 \pm 0.015$~&~$0.358 \pm 0.008$~&~$0.267 \pm 0.006$~\\
~Enron (W)~&~$0.927 \pm 0.027$~&~$0.591 \pm 0.021$~&~$0.400 \pm 0.012$~\\
~Enron (D)~&~$2.807 \pm 0.113$~&~$1.851 \pm 0.047$~&~$1.888 \pm 0.033$~\\
~Irvine~&~$21.28 \pm 0.265$~&~$8.151 \pm 0.285$~&~$4.155 \pm 0.107$~\\
~Yahoo $10^4$~&~$3,524 \pm 16.11$~&~$1,518 \pm 16.29$~&~$632.2 \pm 14.28$~\\
~Yahoo $10^5$~&~$\approx 80$ hours~&~$\approx 40$ hours~&~$\approx 20$ hours~\\ \hline
\end{tabular}
 \caption{Runtime evaluation (seconds).}\label{runtime}
 \vspace{-0.1in}
\end{table}

We conclude this section with the runtime evaluation of our framework, as reported in Table~\ref{runtime}. Note that the code of \algoname{DegreeAnonymization} can be easily parallelized using standard multiprocessing programming APIs such as OpenMP\footnote{http://www.openmp.org}. Thus, the runtimes are measured by executing the $l$ outer iterations of \algoname{DegreeAnonymization} in parallel on a server equipped with two 6 cores Intel Xeon E5645 (2.40GHz) HyperThreading enabled CPUs with a total of 24 logical cores and 48GB of RAM. As we can see, the anonymization of the Yahoo graph is the most expensive one in terms of time, as it took approximately $80$ hours to create a $k=2$ anonymous graph. Recall, however, that although the number of nodes of the Yahoo graph is $100,000$, we are effectively operating on a total of $100,000 \times 28 \mbox{(days)} = 2,800,000$ million nodes. Another important observation is that the runtime does not grow linearly with the longitudinal dimension. While Enron Day has $\approx 30$ times the temporal slices of Enron Month, the runtime is only $\approx 5$ times higher. Finally, we should stress that our anonymization technique is aimed at preventing attacks on a time-varying graph that has been published in its entirety, and thus the computational time can be considered a less stringent constraint than the level of privacy that we can ensure.

\section{Discussion and Conclusions}\label{conclusions}

In this paper we have presented a novel framework for the anonymization of time-varying and multi-layer graphs. We have considered the case of an attacker that has access to the number of social ties of an OSN user over time or over multiple online platforms. In order to protect the nodes' identity, we have proposed to perturb the structure of the time-varying graph so that the temporal degree sequence of each node become indistinguishable from that of at least other $k-1$ nodes. To this end, we have introduced a variant of $k$-means in the $l_1$ space with the additional constraint that each group needs to contain at least $k$ nodes. We have also shown how to approximate the problem of making a degree sequence realizable as an iterated linear program, and we have proposed a fast and effective algorithm to solve it. We have applied the proposed framework on a number of real-world and synthetic networks, and we have shown that the amount of edge insertions or deletions that we need to perform depends on the average temporal correlation~\cite{clauset2012persistence} of the graph. In order to evaluate the structural information loss after anonymization is applied, we have compared the PageRank vectors of the original and anonymized temporal slices, and we have found that the amount of structural information that is preserved is higher in temporal slices that correspond to high activity periods.

Note that in our framework there is a clear trade-off between the computational complexity of the proposed algorithms and the quality of the local optimum we converge to. In other words, a lower runtime inevitably comes at the cost of an increased edit distance between the original graph and the anonymized one. We plan to investigate more efficient heuristics, both in terms of time complexity and quality of the solution. In particular, our aim is to modify the proposed framework to be able to scale up to even larger graphs. Another interesting direction of research is the analysis of
scenarios in which the graph is only partially $k$-anonymous, i.e., only a subset of the nodes satisfies $k$-anonymity, or where the level of anonymity $k$ varies across the nodes~\cite{chester2012anonymizing}.

\section*{Acknowledgement}
This work was supported through the ``The Uncertainty of Identity: Linking Spatiotemporal Information Between Virtual and Real Worlds'' Project (EP/ J005266/1), funded by the EPSRC, and the ``LASAGNE'' Project, Contract No. 318132 (STREP), funded by the European Commission.

\small{
\bibliographystyle{aaai}
\bibliography{biblio}
}

\end{document}